\documentclass[journal]{IEEEtran}

\usepackage{longtable,supertabular}
\usepackage{amsmath,algorithm,algpseudocode}
\usepackage{amssymb}
\usepackage{latexsym}

\newtheorem{definition} {{Definition}}
\newtheorem{example} {{Example}}

\newtheorem{theorem} {{Theorem}}
\newtheorem{proposition} {{Proposition}}
\newtheorem{corollary} {{Corollary}}
\newtheorem{proof} {{Proof}}

\newcommand{\cP}{{\cal P}}

\newcommand{\sP}{\cP}
\newcommand{\Ps}{\smash{{\sP\kern-2.0pt}_q\kern-0.5pt(n)}}

\newcommand{\Gr}{\mathcal{G}_q(k,n)}

\DeclareMathOperator{\im}{im}

\newcommand{\gaussmset}[2]{\left[\begin{smallmatrix}{#1}\\{#2}\end{smallmatrix}\right]}
\newcommand{\gaussmnum}[3]{\gaussmset{#1}{#2}_{#3}}

\title{Construction of Constant Dimension Code from Two Parallel Versions of Linkage Construction}
\author{Xianmang He}

\begin{document}
\maketitle
\begin{abstract}
The linkage construction and its generalization  is one of the most powerful constructions for constant dimension code, accounting for approximately 50\% of all the listed parameters.
We show how to improve the linkage  construction of subspace codes by two parallel versions of the linkage construction. 
This proof allows us to attain codes of larger size for a given minimum distance, which exceeds the latest improvements 
on the linkage construction \cite{combining2019Antonio} in the cases $A_q(13,4,4),A_q(17,4,4), A_q(19,6,6)$.

\textbf{keywords:} Linkage construction, Constant dimension codes, Lifted MRD code
\end{abstract}

\section{Introduction}

Let $q>1$  be a prime power, $\mathbb{F}_q$ the field with $q$ elements.  Let $V \cong \mathbb{F}_q^n$ be a $n$-dimensional vector space over the finite field $\mathbb{F}_q$.
We denote the set of all $k$-dimensional subspaces in $V$ with $\Gr$. Its cardinality can be calculated by the $q$-binomial 
coefficient $\gaussmnum{n}{k}{q}=\prod_{i=0}^{k-1}\frac{q^n-1}{q^k-1}$ for $0 \le k \le n$ and $0$ otherwise.
R. K\"otter and F. R. Kschischang \cite{koetter2008coding} have proven that the set $V$ endowed with the distance $d$ defined by
\begin{align*}
\forall U,W \in V, d(U, W) &= \dim(U + W) - \dim(U\cap W)\\
&= 2 \dim(U+W)  - \dim(U)-\dim(W), 
\end{align*}
is a metric space.

Subspace coding, equipping with this metric space,  was  first applied for error control and correction  in random linear network coding by the pioneering work \cite{koetter2008coding}. 
Specially,  an $(n,M,d,k)_q$ constant dimension code (CDC) $C$ is a subset of $\Gr$  with the size $M$ in which for all $W \ne U \in C$, we always have $d_S(U,W)\ge d$.
In other words, the subspace distance is lower bounded by $d$: $d_S(C)\ge d$.
The main problem of constant dimension subspace coding is to explore the maximum value $M$ of a $(n,M,d,k)_q$ code under the fixed parameters $q$, $n$, $d$, and $k$.
In general, the maximum possible size $M$ is often denoted by $A_q(n,d,k)$.

The most powerful construction is the linkage construction \cite{Heide2016linkage} and its generalization\cite{heinlein2017asymptotic}.
In 2016, Heide and Carolyn \cite{Heide2016linkage} proposed  a construction coming from Corollary 39 in \cite {silberstein2014subspace}.  
Later, an improved linkage construction was presented, and the following lower bound was found:
$A_q(n,d,k)\ge A_q(m,d,k)q^{max\{n-m,k\}(min\{n-m,k\}-\frac{d}{2}+1)}+A_q(n-m+k-\frac{2}{d},d,k).$
Fagang Li~\cite{FagangLi2019Construction} combined the linkage construction and echelon-Ferrers to obtain new lower bounds for constant-dimension codes and
 improved the linkage construction in some cases. 
 Sascha Kurz\cite{Kurz2019}  generalized the linkage construction, and this comes at the cost of introducing a new notation $B_q(v_1,v_2,d,k)$.
Recently, Sascha Kurz  gave some algorithmic results in  the  cases $d=4$ \cite{lifted2020kurz}.

In this paper, we present  a new construction for constant dimension  codes from two parallel versions linkage construction. 
Some new constant-dimension dimension subspace codes of larger size and the expression of these bounds  are also given.

\section{Previous Known Results}
Generally, the exact value of $A_q(n,d,k)$ is a hard problem both algorithmically and theoretically.  Even in the case when the parameters are relatively small. 
As yet, there are only three non-trivial cases of constant dimension codes that the maximum number of codewords have been determined.
They are $A_2(6,4,3)$=77\cite{Honold2013Optimal}, $A_2(8,6,4)=257$\cite{Heinlein2017Classifying} and  $A_2(13,4,3)$=1597245\cite{etzion2013problems}, 
while other non-trivial parameters need further exploration.
A plethora of  results on the construction of CDCs are invented in the literatures.  The upper and lower  bounds on $A_q(n,d,k)$ have been in-depth investigated in the last decade.
 The report \cite{Heinlein2016Tables}  depicts  an on-line database, to which we refer the online tables in the website http://subspacecodes.uni-bayreuth.de.
 Tables list the cases  including  $q\in \{2,3,4,5,7,8,9\}$ and $n$ varies from 4 to 19,  gathering the state-of-art information about the known upper and lower bounds for  constant dimension subspace codes.

New subspace codes from two parallel versions of lifted maximum rank distance codes were introduced  by Xu and Chen\cite{XuChen}. 
This construction yields a  lot of follow-up work  \cite{combining2019Antonio,HaoXJL2019,generalized2019heinlein}. 
Among them, \cite{combining2019Antonio} explores several approaches to combine subspace codes with $k$-spread, which improve on the lower bounds for CDCs for many cases, including $A_q(12,4,4)$, $A_q(12,6,6)$ and $A_q(16,4,4)$, etc.
Geometric concepts such as the Veronese variety and the Segre variety  were also applied to get some new lower bounds in the cases $A_q(2n,4,n)$ \cite{Cossidente2014Subspace}.

Another outstanding construction for constant dimension code uses maximum rank distance (MRD) codes, see the section \ref{sec-liftmrd}. The expurgation-augmentation method was invented by Thomas Honold \cite{Honold2013Optimal},
which starts from a lifted MRD code and then adds and removes some special codewords. With this,  Thomas Honold etc. finally determined the maximum size $A_2(6,4,3)$ to 77. 
The echelon-Ferrers construction \cite{etzion2009error} is a good method among the subspace distance, the rank distance and the Hamming distance, and it is suitable for construction under various parameters.
A greedy-type algorithm for echelon-Ferrers construction has been proposed by Alexander Shishkin, see \cite{Shishkin2014,shishkin2014cardinality}.   
In \cite{gabidulin2011new,multicomponent} the authors studied the block designs to improve the echelon-Ferrers construction.

\section{Construction}

\subsection{Lifted MRD code}\label{sec-liftmrd}
A linear rank metric code $[m \times n, M, d]_q$ is a subspace $C$ of the vector space of $m \times n$ matrices over $\mathbb{F}_q$, i.e., $\mathbb{F}_q^{m \times n}$, of the size $M$, for which the distance of each pair of 
elements is lower bounded through the rank metric $d_r(A,B) = rank(A-B)$. For all suitable parameters, $m, n, d>0$ and prime power $q$, there exists a linear rank metric code 
that reaches the maximum size of $\left\lceil q^{\max\{m,n\}(\min\{m,n\}-d+1)} \right\rceil$. We denote this with $Q_q(m,n,d)$.

The lifted MRD (LMRD) code \cite{silva2008rank} is a $(n,M,d,k)_q$ CDC $C$ that uses an identity matrix $I_k$(with the size $k \times k$) as the MRD code $Q_q(k,n-k,\frac{d}{2})$ prefix, 
which implies $M=q^{max\{(n-k),k\}(min\{n-k,k\}-\frac{d}{2}+1)}$: $C=\{ \operatorname{rowspan}(I_k \mid A) : A \in Q_q(k,n-k,\frac{d}{2}) \}$, 
where ``$\mid$'' denotes the horizontal concatenation of two matrices having the same number of rows.

\subsection{Two parallel versions of linkage construction}\label{linkage_construction}

In this section, we first look back some basic definitions in linkage construction \cite{Heide2016linkage}. 

\begin{definition}[\cite{Heide2016linkage}]
A set ${\bf U} \subset \mathbb{F}_q^{k \times n}$ with the size $k \times n$ matrices over $\mathbb{F}_q$ is called a SC-representation of 
a set of $k$ dimensional subspaces in $\mathbb{F}_q^n$ such that for all $U \in {\bf U}$,  $rank(U)=k$,    and for all $U_1 \neq U_2$ in ${\bf U}$, we have $Im(U_1)\neq Im(U_2)$ . 
\end{definition}

Here $Im(U)$ is the $k$ dimensional subspace spanned by $k$ rows of $U$.\\

\begin{proposition}(see \cite{Heide2016linkage}). {\em  Let ${\bf U}$ be a SC-representation of a $(n_1, N_1, d_1, k)_q$ constant dimension code 
and ${\bf Q} \subset Q_q(n_2,k,d_2)$  be a code with $N_2$ elements and rank distance $d_2$. 
Note that  the set of $k$ dimensional subspaces in ${\bf F}_q^{n_1+n_2}$ defined by ${\bf W_1}=\{Im(U_{11}|Q_{12}): U_{11} \in {\bf U}, Q_{12} \in {\bf Q}\}$. This is a $(n_1+n_2, N_1N_2, \min\{d_1,2d_2\}, k)_q$ constant dimension code. Here $(U|Q)$ is a $k \times (n_1+n_2)$ matrix concatenated from $U$ and $Q$.} \label{proposition-linkage} \\
\end{proposition}

Similar to code $W_1$, we have another ($n_1+n_2, N_3N_4, \min\{d_1,2d_2\},k)_q$ constant dimension code $W_2=\{Im(Q_{21}|U_{22}): Q_{21} \in {\bf Q} , U_{22} \in \bf{U} \}$, where
$Q \subset Q_q(n_1,k,d_2)$  be a MRD code with rank distance $d_2$ and $N_3$ elements, ${\bf U}$ be a SC-representation of a $(n_2, N_4, d_1, k)_q$ constant dimension subspace code.

Now the problem is how many different subspaces we can take from these two parallel versions of linkage construction so as to preserve
the subspace distance $d$.

\subsection{Delsarte Theorem}

The rank distribution of a code ${\bf Q}$ in $Q_q(m,n,d)(m\ge n)$ is defined by $A_r({Q_q(m,n,d)})=|\{Q \subseteq {Q_q(m,n,d)}, rank(Q)=r\}|$ for $r \in  [n,m]$ (see \cite{Delsarte1978Bilinear,Cruz2015Rank}).
The rank distribution of a MRD code is completely  determined by its parameters. 
Such result can be referred to  Theorem 5.6 in \cite{Delsarte1978Bilinear} or Corollary 26 in \cite{Cruz2015Rank}. 
The Delsarte Theorem is used to calculate the final result in this paper.

\begin{theorem}\label{them-delsarte}
{\bf(Delsarte 1978)} {\em Assume that ${\bf Q} \subseteq {Q_q(m,n,d)} $ ($m\ge n$)is a MRD code with rank distance $d$, then its rank distribution is given by 
$$A_r(Q_q(m,n, d))=\displaystyle{n \choose r}_q \Sigma_{i=0}^{r-d} (-1)^i q^{\displaystyle{i \choose 2}} \displaystyle{r \choose i}_q (\frac{q^{m(n-d+1)}}{q^{m(n+i-r)}}-1),$$ 
where $d \le r \le n$, $A_r(Q_q(m,n,d))$ denotes the cardinality of code $Q_q(m,n,d)$ with rank $r$.}
\end{theorem}

\begin{example}\label{exam-delsate}
 Suppose that $m=10, n=4, q=2, d=2$, we have $|Q_2(10,4,2)|=2^{30}$, and $A_2(Q_2(10,4,2))=35805, A_3(Q_2(10,4,2))=15621210,A_4(Q_2(10,4,2))=1058084809$. 
 Similarly, $m=8, n=4, d=2, q=2$, then $|Q_2(8,4,2)|=2^{24}$, and $A_2(Q_2(8,4,2))=8925, A_3(Q_2(8,4,2)))=956250, A_4(Q_2(8,4,2))=15812040$.
\end{example}

\subsection{A new lower bound for $A_q(n_1+n_2, d, k)$}
In this section, we give our main construction in the following theorem \ref{theo-main} and theorem \ref{theo-main-2}.

\begin{theorem}\label{theo-main}
{\em Let ${\bf U}$ and ${\bf V}$ be two SC-representations of   $(n_1, N_1, d, k)_q$ and $(n_2, N_3, d, k)_q$ constant dimension codes, respectively. 
Let ${\bf Q}_1 \subset Q_q(n_2,k,\frac{d}{2})$ be a code with $N_2$ elements and rank distance $\frac{d}{2}$. 
Let ${\bf Q}_2 \subset Q_q(n_1,k,\frac{d}{2})$ be a code with $N_4$ elements and rank distance $\frac{d}{2}$  such that the rank of each element in ${\bf Q}_2$ is at most $k-\frac{d}{2}$. 
Then we have a $(n_1+n_2, N_1N_2+N_3N_4, d, k)_q$ constant subspace code.}\\
\end{theorem}

\begin{proof}Consider the code  ${\bf C}=\{\im(U_{11}|Q_{12}): U_{11} \in {\bf U}, Q_{12} \in {\bf Q}_1\} \cup \{\im(Q_{21}|U_{22}): Q_{21} \in  {\bf Q_2}, U_{22} \in {\bf V}\}.$
From the Proposition \ref{proposition-linkage}, we know that ${\bf W}_1=\{\im(U_{11}|Q_{12}): U_{11} \in {\bf U}, Q_{12} \in {\bf Q}_1\}$ 
and ${\bf W}_2=\{\im(Q_{21}|U_{22}): Q_{22} \in {\bf Q}_2, U_{22} \in {\bf V}\}$ are two parallel versions of linkage construction. Therefore, these two codes are disjoint.
We need to prove that the subspace distance between $W_1 \in {\bf W}_1$ and $W_2 \in {\bf W}_2$ is at least $d$. \\

It is sufficient to prove that
$$dim(W_1+W_2)=  rank\left(
\begin{array}{cccc}
U_{11} &  Q_{12} \\
Q_{21}& U_{22}
\end{array} \right) \ge k+\frac{d}{2}$$

We can exchange columns in the first $n_1$ columns to make  the front $k$ columns in $U_{11}$  be a $k\times k$ unit matrix $E_k$:   $U_{11}=\{E_k, U_{11}'\}$, $U_{11}'$ is a matrix with $(n_1-k)\times k$.
In the meanwhile, $Q_{21}$ will be transformed to $Q_{21}'=\{Q_{211}, Q_{212}\}$,  where $Q_{211}$ is a matrix with $k\times k$, and $Q_{212}$ is a matrix with $(n_1-k)\times k$. Then,

$$dim(W_1+W_2)=  rank\left(
\begin{array}{cccc}
E_k & U_{11}' &  Q_{12} \\
Q_{211} & Q_{212} & U_{22}
\end{array} \right) $$

The above formula can be transformed into the following  by subtracting first row multiplied by $Q_{211}$:

$dim(W_1+W_2) $
$$= rank\left(
\begin{array}{cccc}
E_k & U_{11}' &  Q_{12} \\
0 & Q_{212}-Q_{211}\times U_{11}' & U_{22}-Q_{211}\times Q_{12}
\end{array} \right) $$

Consider that $rank(U_{22}-Q_{211}\times Q_{12}) \ge rank(U_{22})-rank(Q_{211}\times Q_{12}) \ge rank(U_{22})-rank(Q_{12})=k-(k-\frac{d}{2})=\frac{d}{2}$.
The conclusion is proved.\\

\end{proof}

We utilize  theorem \ref{them-delsarte}, and   give a concrete calculation formula  of the theorem  \ref{theo-main} in the following corollary.

\begin{corollary}\label{coro-main}
If $k\ge d, n_1\ge k, n_2\ge k$, then we have $A_q(n_1+n_2,k,d)\ge |Q_q(n_1,k,\frac{d}{2})| \times A_q(n_2,k,d)$+$A_q(n_1, k,d)  \times (1+\sum_{r=\frac{d}{2}}^{k-\frac{d}{2}} A_r(Q_q(n_2,k,\frac{d}{2}))).$
\end{corollary}

The following example is used to illustrate the corollary \ref{coro-main}.

\begin{example}
Let $k=4, d=4, n_1$ is fixed as 8, and $n_2$ varies from 4 to 11, from corollary \ref{coro-main}, we have
 $A_q(n_1+n_2, 4,4) \ge A_q(8,4,4) \times |Q_q(n_2,4,2)|+(1+A_2(Q_q(8,4,2)))\times A_q(n_2,4,4).$
When $n_2=4$, $A_q(12,4,4)\ge A_q(8,4,4) \times |Q_q(4,4,2)|+1+A_2(Q_q(8,4,2))$, here $A_q(4,4,4)$ will be degenerated to $4\times 4$ identity matrix $I_4$.
Assume that $q=2$, $A_2(12,4,4)\ge A_2(8,4,4) \times |Q_2(4,4,2)|+A_2(Q_2(8,4,2))=4801\times 4096+8925=19673822$.
\end{example}

From the proof of the theorem \ref{theo-main}, we notice that  ${\bf Q}_2 \subset Q_q(n_2+t,k,\frac{d}{2})$(where $ 0\le t \le n_1-k $) is a code with rank distance $\frac{d}{2}$ such that the rank of each element in ${\bf Q}_2$ is at most $k-\frac{d}{2}$ ,
the theorem is still true.   More general, we have the following theorem.

\begin{theorem}
\label{theo-main-2}
If $k\ge d, n_1\ge k, n_2\ge k, 0\le t \le n_1-k $, then we have $A_q(n_1+n_2,k,d)\ge |Q_q(n_1,k,\frac{d}{2})| \times A_q(n_2,k,d)+A_q(n_1-t, k,d)  \times (1+ \sum_{r=\frac{d}{2}}^{k-\frac{d}{2}} A_r(Q_q(n_2+t,k,\frac{d}{2}))) .$
\end{theorem}

\begin{proof}
Suppose  $U$ and $V$ be  two SC-representations of  $(n_1, N_1, d, k)_q$ and $(n_2-t, N_3, d, k)_q$ constant dimension codes, respectively. 
Let ${Q}_1 \subset Q_q(n_2,k,\frac{d}{2})$ be a code with rank distance $\frac{d}{2}$,  
 ${Q}_2 \subset Q_q(n_2+t,k,\frac{d}{2})$ be a code with rank distance $\frac{d}{2}$ and the rank of each element in ${\bf Q}_2$ is at most $k-\frac{d}{2}$, where $0\le t \le n_1-k$.

We define the code as ${\bf C}=\{\im(U_{11}|Q_{12}): U_{11} \in {\bf U}, Q_{12} \in {\bf Q}_1\} \cup \{\im(Q_{21}|U_{22}): Q_{21} \in  \bf{Q_2}, U_{22} \in {\bf V}\}.$

Consider that ${\bf W}_1$ and ${\bf W}_2$ are two versions of linkage construction, therefore, it is clear that  the subspace distances of these two codes ${\bf W}_1=\{\im(U_{11}|Q_{12}): U_{11} \in {\bf U}, Q_{12} \in {\bf Q}_1\}$
 and ${\bf W}_2=\{\im(Q_{21}|U_{22}): Q_{22} \in {\bf Q}_2, U_{22} \in {\bf V}\}$  themselves are at least $d$. 
Hence, we need to prove that the subspace distance between $W_1 \in {\bf W}_1$ and $W_2 \in {\bf W}_2$ is at least $d$. \\

Consider that the dimension of $W_1\cap W_2=\{x(U_{11}|Q_{12})=y(Q_{21}|U_{22}),x,y\in \bf{F}_q^k\}$, hence, $xU_{11}=yQ_{21}'$, $Q_{21}'$  is  the first $n_1$ column of the  matrix $Q_{21}$.
We notice that there exists an identity matrix $E$ with size $k\times k$ in  $U_{11}$. No matter where this identity matrix the position is,  
we always have that the dimension of the subspace  $\{x: \exists y, xE=yQ' \}$ is  at most the rank of the matrix $Q'$, that is $k-\frac{d}{2}$, the matrix $Q'$ is the corresponding matrix of $E$ in $Q_{21}$.  Then

$$d(W_1,W_2)\ge 2k-2(k-\frac{d}{2})=d.$$

This completes the proof.

\end{proof}

Now,we present some examples to illustrate the theorem \ref{theo-main-2}.

When $n_1=13, k=6, d=6, n_2=6, t=1$, we have $A_q(19,6,6)\ge A_q(12,6,6)\times |Q_q(6,6,3)|+(1+A_2(Q_q(13,6,3))$.
Assume that $q=2$,  we have $A_q(19,6,6)\ge 16865630\times 2^{15}+1+11426445=4527333091203726$.
This bound is strictly improves upon the corresponding results in \cite{generalized2019heinlein,HaoXJL2019,heinlein2017asymptotic,combining2019Antonio}.

When $n_1=8, k=4, d=4, n_2=5, t=1$, we have $A_q(13,4,4)\ge A_q(8,4,4)\times |Q_q(5,4,2)|+(1+A_2(Q_q(9,4,2))$.
Assume that $q=2$,  we have $A_q(13,4,4)\ge 4801\times 2^{15}+1+17885=157337054$,  which exceeds the current best theoretic bound 157332190.

When $n_1=12, k=4, d=4, n_2=5, t=1$, we have $A_q(17,4,4)\ge A_q(12,4,4)\times |Q_q(5,4,2)|+(1+A_2(Q_q(13,4,2))$.
Assume that $q=2$,  we have $A_q(17,4,4)\ge 19676797\times 2^{15}+1+286685=644769570782$, which exceeds the current best theoretic bound 644769492958.

These new bounds exceed the current theoretic  bounds in \cite{generalized2019heinlein,HaoXJL2019,heinlein2017asymptotic}, even the latest improvements in \cite{combining2019Antonio}.
\cite{lifted2020kurz} gives some new algorithmic bounds in the case $d=4$.

\subsection{On the case $d > k$}
The drawback of the construction in theorem \ref{theo-main-2} and theorem \ref{theo-main} is that it is only applicable for the case $d \le  k$.  
If we use the notation of the rank-restricted rank-metric code (RRMC)\cite{generalized2019heinlein}, we can construct CDCs for $d > k$. 

\begin{definition}[\cite{generalized2019heinlein}]
A rank-metric code (RMC) is a subset of $\mathbb{F}_q^{m \times n}$ with cardinality $N$ such that the rank distance $d_r(A,B)\ge d$, for all $A \neq B$.
Additionally, if the rank of each codeword is at least $u$, we use the notation $(m \times n,N,d;u)_q$ to denote it and call it rank-restricted RMC (RRMC).
\end{definition}

The maximum size of an $(m \times n,N,d;u)_q$ RRMC is denoted as $\Lambda(q,m,n,d,u)$.  With this, we have the following corollary. 

\begin{corollary}
\label{coro-3}
If $d \ge k, n_1\ge k, n_2\ge k, 0\le t \le n_1-k $, then we have $A_q(n_1+n_2,k,d)\ge |Q_q(n_1,k,\frac{d}{2})| \times A_q(n_2,k,d)+A_q(n_1-t, k,d)  \times  \Lambda(q,k,n_2+t,\frac{d}{2},k-\frac{d}{2}).$
\end{corollary}

\begin{proof}
We define the code as ${\bf C}=\{\im(U_{11}|Q_{12}): U_{11} \in {\bf U}, Q_{12} \in {\bf Q}_1\} \cup \{\im(Q_{21}|U_{22}): Q_{21} \in  {\bf Q_2}, U_{22} \in {\bf V}\}.$
Let $U$ and $V$ be two SC-representations of  $(n_1, N_1, d, k)_q$ and $(n_2-t, N_3, d, k)_q$ CDCs, respectively. 
Let ${Q}_1 \subset Q_q(n_2,k,\frac{d}{2})$ be a code with rank distance $\frac{d}{2}$, ${Q}_2 \subset ((n_2+t)\times k,N,\frac{d}{2};k-\frac{d}{2})_q$ be a rank-restricted RMC code with rank distance $\frac{d}{2}$ such that the rank of each element in ${\bf Q}_2$ is at most $k-\frac{d}{2}$, where $0\le t \le n_1-k$.
Similar to the proof of theorem \ref{theo-main-2},
It is clear that there exist an identity matrix with the size $k\times k$ in $U_{11}$, and the rank of the corresponding matrix in $Q_{21}$ is at most  $k-\frac{d}{2}$. 
Therefore, the dimension of  $U\cap V$ is at most $k-\frac{d}{2}$.
\end{proof}

Remark: Independent to this paper, Heinlein proposed a variation of the generalized linkage
construction \cite{generalized2019heinlein}.  When $t=0$, the two constructions are exactly the same. 
When $t>0$, the second part of the lower bound is different. 
In addition, the generalized linkage construction \cite{generalized2019heinlein} will degenerate into the improved\_linkage \cite{heinlein2017asymptotic} when $t>0$.

\section{Conclusion}
In this paper, we propose a new construction for constant dimension codes from two parallel versions of linkage construction, and its variants.  
 This construction gives an improved bounds for linkage construction when $k\ge d$.
 In addition, the notation of the rank-restricted rank-metric code (RRMC) is applied to construct bounds for $d\ge k$.  
 We have improved at least the following lower bounds: $A_q(13,4,4), A_q(17,4,4), A_q(19,6,6)$ , and the expression of these bounds are also given. 
 All these theoretic bounds exceeds the bounds presented in \cite{HaoXJL2019,heinlein2017asymptotic,combining2019Antonio,generalized2019heinlein}.


%
%
%

\end{document}